\documentclass[a4paper,UKenglish]{lipics-v2016}

\usepackage{microtype}

\usepackage{epsfig}
\usepackage{amssymb}
\usepackage{amsmath}
\usepackage{ amsfonts} 
\usepackage{stmaryrd}
\usepackage{mathrsfs}
\usepackage{latexsym}
\usepackage{hhline}
\usepackage{gastex}
\usepackage{graphicx}

\title{ { Construction of  Non-expandable Non-overlapping Sets of Pictures}
\footnote{Partially supported by MIUR Projects
{\it ``Automata and Formal Languages: Mathematical and Applicative Aspects''} and {\it ``PRISMA PON04a2 A/F''}, and
by 
FARB Projects of
Universities of Catania, Roma ``Tor Vergata'', Salerno.
 }
}

\titlerunning{Construction of non-expandable non-overlapping sets of pictures}

\author[1]{Marcella Anselmo}
\author[2]{Dora Giammarresi}
\author[3]{ Maria Madonia}

\affil[1]{ Dipartimento di Informatica,
Universit\`a di Salerno, Via Giovanni Paolo II, 132-84084 Fisciano (SA) Italy. E-mail: {\tt
anselmo@dia.unisa.it}}
\affil[2]{  Dipartimento di Matematica.
 Universit\`a  Roma ``Tor Vergata'',
  via della Ricerca Scientifica, 00133 Roma, Italy. E-mail:
{\tt giammarr@mat.uniroma2.it}}
\affil[2]{Dipartimento di Matematica e
Informatica, Universit\`a  di
  Catania, Viale Andrea Doria 6/a,  95125 Catania, Italy.
  E-mail: {\tt madonia@dmi.unict.it}}

\authorrunning{M. Anselmo, D. Giammarresi and M. Madonia}

\Copyright{M. Anselmo, D. Giammarresi and M. Madonia}


\begin{document}

\thispagestyle{empty}

 \maketitle

\begin{abstract}
The non-overlapping sets of pictures are sets such that no two pictures in the set (properly) overlap. They are the generalization  to two dimensions of the cross-bifix-free sets of strings.
Non-overlapping sets of pictures are non-expandable when no other picture can be added without violating the property.
We present a construction of non-expandable non-overlapping (NENO) sets of pictures and show some examples of application.
\\
\end{abstract}

 \section{Introduction}\label{s-intro}

The digital technology that pervades every aspect of our lives is bringing communications more and more towards pictorial (two-dimensional) environments. The generalization to two dimensions of the formal study of all  structures and special patterns of the strings is then gaining a growing  interest in the scientific community.
The two-dimensional strings are called \emph{pictures} and they are represented by two-dimensional (rectangular)  arrays over a finite alphabet $\Sigma$. The set of all pictures over $\Sigma$ is usually denoted by $\Sigma^{**}$.
Extending results from the formal (string) language theory to two dimensions is   a very  challenging task. The two-dimensional structure in fact imposes some intrinsic difficulties  even in the basic concepts. For example,  we can define  two concatenation operations (horizontal and vertical concatenations) between two pictures, but they are only partial operations and do not induce a monoid structure to the set $\Sigma^{**}$. Moreover, for example, the definition of ``prefix" of a string can be extended to a picture by considering its rectangular portion in the top-left corner; nevertheless, if one deletes a prefix from a picture, the remaining part is not a picture anymore.

Several results from string language theory have been worthy extended to pictures. Many researchers
have investigated how the notion of recognizability by finite state automata
  can be  transferred to two dimensions to accept picture languages (\cite{AGM-Fund10,AGMR,AGM-IJFCS14,BH67,GR92,GR-book}).

A relevant notion on strings is the one of "border".
  Given a string $s$, a \emph{bifix} or a \emph{border} of $s$ is a  substring $x$ that is both prefix and suffix of $s$.  A string $s$ is \emph{bifix-free} or \emph{unbordered} if it has no other bifixes,  besides the empty string and $s$ itself.
Bifix-free strings are  strictly related  with the theory of codes \cite{BPR} and are involved in the data structures for pattern matching algorithms \cite{CR02,Gus97}.  From a more  applicative point of view, bifix-free strings are suitable as
 synchronization patterns in digital
communications and similar communication protocols \cite{N73}.
In this framework, the \emph{cross-bifix-free codes} have been introduced in \cite{Bajic07}.
A set of strings $X$  is  a cross-bifix-free code
when no
prefix of any string is the suffix of any other string in $X$; 
it is {\em non-expandable} if
 no other element can be added to $X$ without falsifying the property of the set.
Several efforts have been made in the last years in order to construct families of non-expandable  cross-bifix-free codes.
A first family has been exhibited in \cite{Bajic14}. Subsequently, other families based on different approaches have been shown in \cite{BPP12,CKPW13},
with the aim of finding codes of bigger cardinality.




The notion of border extends very naturally from strings to pictures since it is not related to any scanning direction.
 Informally, we can say that a picture $p$ is \emph{bordered}
if a copy $p'$ of $p$ can be   overlapped on $p$  by putting a corner of $p'$ somewhere on some position in $p$.
 Observe that, depending on the position of the corner of $p'$, several   different types of picture overlaps are possible and some of them can be studied by reducing them to the string case.
 We stay in the  general situation when the overlaps can be made on any position in $p$  and therefore the borders can be of any size. This leads to a quite different scenario with respect to the string case. In fact, we have two pairs of symmetric cases; either $p$ can be overlapped by putting its top-left corner in a copy $p'$ (i.e. the bottom-right corner of $p'$ is inside $p$) or we can put the bottom-left corner of $p$ somewhere in some positions of its copy $p'$  (i.e. the top-right corner of $p'$ is inside $p$).
 Unbordered pictures, in this general setting, were first investigated in \cite{AGM-CAI15}, where an algorithm for the  construction of the set of all unbordered pictures of a fixed size is proposed.

 Moving from pictures to sets of pictures,  we consider the notion of \textit{non-overlapping set of pictures}.  In a non-overlapping set of pictures, each picture is unbordered and moreover no picture can be overlapped on another one of the same set. This notion gives therefore the generalization of cross-bifix-free code of strings and provides families of picture codes. An example of a non-overlapping set of pictures was recently given in \cite{BBP16};   its cardinality is calculated using generalized Fibonacci sequences. Actually, it is quite easy to find "small" non-overlapping sets of pictures of a fixed size $(m,n)$; the challenge is, of course, to find "big" ones! More precisely, the aim is to find "big" non-overlapping sets of pictures which cannot be expanded  and, at the same time, to provide a description for any size of the pictures.
  In the same paper \cite{BBP16}, it is left as main open problem whether it is possible to construct a \emph{non-expandable non-overlapping} set of pictures (we will call it a NENO set, for short). A solution to this problem is a generalization to two dimensions of the above cited results on non-expandable cross-bifix-free sets of strings. 
It constitutes the main contribution of this paper.


First, we show some necessary conditions which are satisfied by any non-overlapping set of pictures.
Subsequently, we identify some conditions which ensure that a set is non-overlapping non-expandable.
The focus is firstly on the properties that the frames of the pictures may have, and then on the internal part of the pictures.
This approach provides a construction in two main steps.
As an application we exhibit a family of NENO sets $Y(m,n)$, for any $m,n\geq 4$.
It is the first example in the literature.

%
%

\section{Preliminaries}\label{s-preli}

In this section  
we recall all the definitions on strings and pictures needed in the rest of the paper. For more details see \cite{BPR} and \cite{GR-book}.

\subsection{Basic notations on strings}\label{ss-string}
A string is a sequence of zero or more symbols
from an alphabet $\Sigma$.
%
Let $s=s_1s_2 \ldots s_n$ be a string of length $n$ over $\Sigma$. The length of $s$ is denoted $|s|$.
%
%
A string $w$ of length $h$, $h\leq n$,  is a substring (or factor) of $s$ if $s = uwv$ for $u, v\in\Sigma^*$. Moreover we say
that  $w$ occurs at position $j$ of $s$  if $w=s_j \ldots s_{j+h-1}$. 
A string $x$ of length $m<n$ is a \emph{prefix} of $s$ if $x$ is a substring that occurs in $s$ at position $1$;
 it is a \emph{suffix} of $s$ if it is a substring that occurs in $s$ at position $n-m+1$.
 A string $x$ that is both prefix and suffix of $s$ is called a  \emph{border} or a \emph{bifix} of $s$.
 The empty string and $s$ itself are \emph{trivial} borders of $s$. A string $s$ is \emph{unbordered} or \emph{bifix-free} if it has no borders unless the trivial ones.
Unbordered strings have received very much attention since they occur in many applications as message synchronization or string matching. In \cite{N73} P. T. Nielsen proposed a  procedure to generate all bifix-free strings of a given length.

Finally, two strings $s$ and $s'$ {\em overlap} if there exists a string $x$ that is a suffix of $s$ and a prefix of $s'$, or vice versa; the string $x$ is their overlap and $|x|$ is the length of the overlap. We will equivalently say that $s$ overlaps $s'$.

\subsection{Basic notations on pictures}\label{ss-2d-notation}
A {\em picture} over a finite
alphabet $\Sigma$ is a two-dimensional rectangular array of
elements of $\Sigma$. Given a picture $p$, $|p|_{row}$ and $|p|_{col}$ denote the number of rows and columns, respectively while $size(p)=\left(|p|_{row},|p|_{col}\right)$ denotes the picture {\em size}. The pictures of size
$(m,0)$ or $(0,n)$ for all $m,n\geq 0$,  called  \emph{empty} pictures, will be never considered in this paper.
The set of all pictures over $\Sigma$ of fixed size $(m,n)$ is denoted by $\Sigma^{m,n}$, while
the set of all pictures over $\Sigma$ is denoted by $\Sigma^{**}$.


Let $p$ be a picture of size $(m,n)$. The  set of coordinates $dom(p)=\{1, 2, \ldots, m\}\times \{1, 2, \ldots, n\}$ is referred to as the \emph{domain} of a
picture $p$.
We let $p{(i,j)}$ denote the symbol  in $p$ at coordinates $(i,j)$.
We assume the top-left corner of the picture to be at position $(1,1)$.
Moreover, to easily detect border positions of pictures, we use initials of words ``top", ``bottom", ``left" and ``right"; then, for example, the
 \emph{tl-corner} of $p$ refers to position $(1,1)$ while the \emph{br-corner} refers to position $(m,n)$.
Furthermore,  we denote by
$r_F(p), r_L(p)\in \Sigma^n$ the first and the last row of $p$, respectively
and
by
$c_F(p), c_L(p)\in \Sigma^m$ the first and the last column of $p$, respectively. Then, the
 {\em frame} of $p$ is
$frame(p)=(r_F(p), r_L(p), c_F(p), c_L(p))$.

For the sequel, it is convenient to  extend the notation for the frame of a picture to languages.
Let $X\subseteq \Sigma^{m,n}$. Let us denote by
$R_F(X)\subseteq\Sigma^n$ the set
$R_F(X)=\{ r_F(p) \ | \ p\in X  \}$ of the first rows of all pictures in $X$.
In a similar way,
$R_L(X)$, $C_F(X)$, and $C_L(X)$
will denote
the sets of the last rows,
 of the first columns,
and
of the last columns of all pictures in $X$, respectively.
%
%
The {\em frame} of $X$ is the quadruple
$frame(X)=(R_F(X), R_L(X), C_F(X), C_L(X))$.

  A {\em subdomain} of $dom(p)$ is
  a set $d$ of the form $\{i, i+1, \ldots, i'\}\times \{j, j+1, \ldots,j'\}$, where
  $1\leq i\leq i'\leq m,\ 1\leq j\leq j'\leq n$,
  also specified  by the pair $[(i, j), ({i'}, {j'})]$.
The portion of $p$ corresponding to positions in subdomain $[(i, j), ({i'}, {j'})]$
 is  denoted by $p[(i, j), ({i'}, {j'})]$.
Then, a non-empty picture $x$ is \emph{subpicture of $p$} if $x=p[(i, j), ({i'}, {j'})]$, for some  $1\leq i\leq i'\leq m,\ 1\leq j\leq j'\leq n$; we say that $x$ {\em occurs} at position $(i,j)$ (its tl-corner).

Observe that the notions of subpicture generalizes very naturally  to two dimensions the notion of substring. On the other hand, the notions of prefix and suffix of a string implicitely assume the left-to-right reading direction.
In two dimensions, there are 4 corners and 4 scanning-directions from a corner toward the opposite one. Hence, we introduce the definition of 4 different "prefixes" of a pictures, each one referring to one corner.


 \begin{definition}\label{d-prefix}
Given pictures $p \in \Sigma^{m,n}$, $x \in \Sigma^{h,k}$,
with $1\leq h\leq m$, $1\leq k\leq n$,

$x$ is a \emph{tl-prefix} of  $p$
 if $x$ is a subpicture of $p$ occurring at position $(1,1)$,

$x$ is a \emph{tr-prefix} of  $p$  if $x$ is a subpicture of $p$ occurring at position $(1,n-k+1)$,

$x$ is a \emph{bl-prefix} of  $p$  if $x$ is a subpicture of $p$ occurring at position $(m-h+1,1)$,

$x$ is a \emph{br-prefix} of  $p$  if $x$ is a subpicture of $p$ occurring at position $(m-h+1,n-k+1)$.
  \end{definition}

Several operations can be defined on pictures (cf. \cite{GR-book}).
Two  concatenation products are usually considered,
 the  {column} and the row concatenation.
%
%
%
The reverse operation on strings can be generalized to pictures and gives rise to  two different mirror operations (called \emph{row}- and \emph{col}-\emph{mirror}) obtained by reflecting with respect to a vertical and a horizontal axis, respectively. Another  operation that has no counterpart in one dimension is the \emph{rotation}.
The rotation of
 a picture $p$ of size $(m,n)$, is the clockwise rotation of $p$ by $90^\circ$, denoted by $p^{90^\circ}$.
Note that $p^{90^\circ}$ has
 size $(n,m)$. 
All the operations defined on pictures can be extended in the usual way to sets of pictures.


%
%

We conclude by remarking that  any string $s=y_1y_2\cdots y_n$  can be identified either with a single-row or with a single-column picture, i.e. a picture of size $(1,n)$ or $(n,1)$,
whereas any picture in $\Sigma^{m,n}$ can be viewed as a string of length $n$ on the alphabet of the columns $\Sigma^m$,
and as  a string of length $m$ on the alphabet of the rows $\Sigma^n$.


\section{Non-overlapping sets of pictures}\label{s-non-ov}

In this section we set up all the necessary definitions and notations on non-overlapping pictures. The notion is strictly related to the one of "unbordered pictures" already introduced and studied in \cite{AGM-CAI15}. Here, we state directly the definition of non-overlapping pictures getting back the notion of unbordered picture as a particular case.

Recall that two strings overlap when the prefix of one of them is the suffix of the other one.
This notion can be extended very naturally   to two dimensions by taking into account that  now 4 different corners exist.
 Informally, we say that two pictures $p$ and $q$ overlap when
%
%
we can find the same rectangular portion at a corner of $p$ and at the opposite corner of $q$.
Observe that there are two different kinds of overlaps  depending on the pair of opposite corners involved.

\begin{definition}
Let $p\in\Sigma^{m,n}$ and $q\in\Sigma^{m',n'}$.
\\
The  pictures $p$ and $q$ {\em tl-overlap} if there exists a picture $x\in\Sigma^{h,k}$, with $1\leq h\leq \min\{m, m'\}$
and $1\leq k\leq \min\{n, n'\}$,
which  is a tl-prefix of $p$ and a br-prefix of $q$, or vice versa.
\\
The  pictures $p$ and $q$ {\em bl-overlap} if there exists a picture $x\in\Sigma^{h,k}$, with $1\leq h\leq \min\{m, m'\}$
and $1\leq k\leq \min\{n, n'\}$,
which  is a bl-prefix of $p$ and a tr-prefix of $q$, or vice versa.
\\
The pictures $p$ and $q$ {\em overlap} if they tl-overlap or they bl-overlap.
\\
The picture $x$ is said an {\em overlap} of $p$ and $q$, and its size  $(h,k)$ is the {\em size of the overlap}.
\end{definition}


For the sequel it will be worthy to identify some special cases of picture overlaps and we list them in the definition below. Examples are given in Figure \ref{fig-overlap} where
the first pair of pictures tl-overlap, the second pair h-slide overlap, the third one v-slide overlap, and the last pair shows two pictures that frame overlap (and also bl-overlap).

\begin{definition}\label{d-ovlp2}
Let $p\in\Sigma^{m,n}$ and $q\in\Sigma^{m',n'}$, then
\begin{itemize}
\item  $p$ and $q$ {\em properly overlap} if they have an overlap $x\in\Sigma^{h,k}$ with $x\neq p$ and $x\neq q$
\item $p$ and $q$ {\em h-slide overlap} if they have an overlap $x\in\Sigma^{h,k}$ with $h=m=m'$
\item $p$ and $q$ {\em v-slide overlap} if they have an overlap $x\in\Sigma^{h,k}$  with $k=n=n'$
\item $p$ and $q$ {\em frame overlap}  if they have an overlap $x\in\Sigma^{h,k}$  with $h=1$ or $k=1$ 
\end{itemize}
\end{definition}


\begin{figure} 

\begin{picture}(120,25)(0,-5)
\put(2.5,11.5){$0$} \put(7.5,11.5){$1$} \put(12.5,11.5){$1$}
\put(2.5,6.5){$1$} \put(7.5,6.5){$1$} \put(12.5,6.5){$0$}
\put(2.5,1.5){$1$} \put(7.5,1.5){$0$} \put(12.5,1.5){$0$}
\put(0,0){\line(1,0){15.5}}
\put(0,15){\line(1,0){15.5}}
\put(0,0){\line(0,1){15}}
\put(15.5,0){\line(0,1){15}}
\put(17.5,6.5){$1$} \put(22.5,6.5){$0$}
\put(17.5,1.5){$1$} \put(22.5,1.5){$0$}
\put(7.5,-3.5){$0$} \put(12.5,-3.5){$1$} \put(17.5,-3.5){$1$} \put(22.5,-3.5){$0$}
\put(5.5,-5){\line(0,1){15}}
\put(25.5,-5){\line(0,1){15}}
\put(5.5,-5){\line(1,0){20}}
\put(5.5,10){\line(1,0){20}}

\put(37.5,11.5){$0$} \put(42.5,11.5){$1$} \put(47.5,11.5){$1$}
\put(37.5,6.5){$1$} \put(42.5,6.5){$1$} \put(47.5,6.5){$0$}
\put(37.5,1.5){$1$} \put(42.5,1.5){$0$} \put(47.5,1.5){$0$}
\put(35,0){\line(1,0){15.5}}
\put(35,15){\line(1,0){15.5}}
\put(35,0){\line(0,1){15}}
\put(50.5,0){\line(0,1){15}}
\put(52.5,11.5){$0$}
\put(52.5,6.5){$0$}
\put(52.5,1.5){$0$}
\put(40,-0.5){\line(1,0){15}}
\put(40,14.5){\line(1,0){15}}
\put(40,-0.5){\line(0,1){15}}
\put(55,-0.5){\line(0,1){15}}

\put(67.5,11.5){$0$} \put(72.5,11.5){$1$} \put(77.5,11.5){$1$}
\put(67.5,6.5){$1$} \put(72.5,6.5){$1$} \put(77.5,6.5){$0$}
\put(67.5,1.5){$1$} \put(72.5,1.5){$0$} \put(77.5,1.5){$0$}
\put(65,0){\line(1,0){15.5}}
\put(65,15){\line(1,0){15.5}}
\put(65,0){\line(0,1){15}}
\put(80.5,0){\line(0,1){15}}
\put(67.5,-3.5){$0$} \put(72.5,-3.5){$1$} \put(77.5,-3.5){$1$}
\put(65.5,10){\line(1,0){15.5}}
\put(65.5,-5){\line(1,0){15.5}}
\put(65.5,-5){\line(0,1){15}}
\put(81,-5){\line(0,1){15}}

\put(92.5,6.5){$0$} \put(97.5,6.5){$0$} \put(102.5,6.5){$1$} \put(107.5,6.5){$1$}
\put(92.5,1.5){$1$}  \put(97.5,1.5){$1$} \put(102.5,1.5){$1$} \put(107.5,1.5){$0$}
\put(92.5,-3.5){$1$}  \put(97.5,-3.5){$1$} \put(102.5,-3.5){$0$} \put(107.5,-3.5){$0$}
\put(90,-5){\line(1,0){20.5}}
\put(90,10){\line(1,0){20.5}}
\put(90,-5){\line(0,1){15}}
\put(110.5,-5){\line(0,1){15}}
\put(107.5,11.5){$0$} \put(112.5,11.5){$1$} \put(117.5,11.5){$1$}
\put(112.5,6.5){$0$} \put(117.5,6.5){$0$}
\put(112.5,1.5){$1$} \put(117.5,1.5){$1$}
\put(105.5,15){\line(1,0){15}}
\put(105.5,0){\line(1,0){15}}
\put(105.5,0){\line(0,1){15}}
\put(120.5,0){\line(0,1){15}}

\end{picture}

\caption{}


\label{fig-overlap}
\end{figure}

The case when $p$ overlaps with itself 
leads to the notions of border of a picture, self-overlapping and unbordered pictures (as investigated in \cite{AGM-CAI15}).
As for the overlaps, there are   two different kinds of borders (\emph{tl-borders} and \emph{bl-border}) depending on the pair of opposite corners that hold the border. A tl-border is called a diagonal border in \cite{CroIK98}.

The notion of "non-overlapping" is naturally extended to sets  of pictures in order  to generalize the notion of cross-bifix-free sets of strings, introduced in \cite{Bajic07}.
Non-overlapping sets of pictures have been introduced and studied in  \cite{BBP16}, where a first example of a non-overlapping set of matrices is presented. Here, we add the notion of "non-expandability" in order to consider sets of large cardinality. Notice that, in analogy to the case of cross-bifix-free sets of strings, we will consider set of pictures of fixed size.

\begin{definition}
A set  of pictures $X\subseteq\Sigma^{m,n}$ is {\em non-overlapping} if for any $p, q\in X$, $p$ and $q$ do not properly overlap.
\\
Moreover, a set $X\subseteq\Sigma^{m,n}$ is {\em non-expandable non-overlapping}, {\em NENO} for short, if $X$ is non-overlapping and
for any $p\in\Sigma^{m,n}\setminus X$, there exists $q\in X$ such that $p$ and $q$ overlap.
\end{definition}


Note that it is not difficult to find a set of non-overlapping pictures of a fixed size $(m,n)$; the real challenge is to find sets of large cardinality and, more generally, to define families \emph{for any}  size $(m,n)$.  

In the  sequel, we propose some conditions for NENO sets that will be used in   Section \ref{s-fanta-neno} to show a family of NENO sets $Y(m,n)$, for any  size $(m,n)$.
Observe that further examples can be simply obtained applying to a NENO set the operations of
rotation, col- and row-mirror operations, and permutation or renaming of symbols in $\Sigma$.

Let us now consider some properties which are necessarily satisfied by any non-overlapping set of pictures.
First, observe that any picture in a non-overlapping set is necessarily unbordered.
In order to show some properties on the frames of the pictures in a non-overlapping set, let us introduce the following definition.
Note that in a cross-non-overlapping  pair $(S_1,S_2)$,  it is not required that $S_1$ and $S_2$ are cross-bifix-free sets.

 \begin{definition}\label{d-full}
Let $S_1$, $S_2 \subseteq \Sigma^n$ and $S_1\cap S_2=\emptyset$.

The pair $(S_1,S_2)$ is {\em cross-non-overlapping} if for any $s_1\in S_1$, $s_2\in S_2$, $s_1$ and $s_2$ do not overlap.\\
\end{definition}

\begin{theorem}
Let $X\subseteq\Sigma^{m,n}$.
If $X$ is non-overlapping then
the pairs $( R_F(X), R_L(X) )$ and $(C_F(X), C_L(X))$ are cross-non-overlapping.
\end{theorem}

\begin{proof}
Suppose by the contrary that there exist $s_1\in R_F(X)$,  $s_2\in R_L(X)$, and $s_1$ and $s_2$ overlap.
Then any pictures $p_1, p_2\in X$ with $r_F(p_1)=s_1$ and 
 $r_L(p_2)=s_2$ are such that $p_1$ and $p_2$ overlap, against $X$ non-overlapping.
An analogous reasoning holds for the sets of columns.
\end{proof}

Some more properties can be stated in the case of the binary alphabet.
If $X$ is a non-overlapping set over a binary alphabet then the four corners of any picture in $X$ carry the same quadruple of symbols.
  We state this simple necessary condition in the lemma below.

Let $corners(p)= (p(1,1), p(1,n), p(m,1), p(m,n))$, for any picture $p$.

\begin{lemma}\label{l-corner-non-ov-sets}
Let $\Sigma=\{ 0,1\}$ and $X\subseteq \Sigma^{m,n}$. If $X$ is a non-overlapping set then
only four cases are possible

a) $corners(p)=(0,0,1,1)$, for any picture $p\in X$

b) $corners(p)=(1,1,0,0)$, for any picture $p\in X$

c) $corners(p)=(1,0, 1,0)$, for any picture $p\in X$

d) $corners(p)=(0,1,0,1)$, for any picture $p\in X$.
\end{lemma}

\begin{proof}
Given an unbordered picture $p$, $corners(p)$ must be of the form a), b), c) or d) otherwise $p$ would have a border of size $(1,1)$.
Since any picture in a non-overlapping set is necessarily unbordered, we have that, for any $p \in X$, $corners(p)$ must be of the form a), b), c) or d).
The proof is completed by noting that, for any $p, q \in X$, it must be $corners(p)=corners(q)$; if the sets $corners(p)$ and $corners(q)$ were of two different forms among
a), b), c) or d), then $p$ and $q$ would have an overlap of size $(1,1)$ against the hypothesis that $X$ is non-overlapping.
\end{proof}










%
%
%
%


\section{The construction}\label{s-constr}
%

In this section we present some properties on non-expandable non-overlapping  sets, from which it is possible to obtain a
a general method to construct  a NENO set of pictures.
The focus will be firstly on the frames of the pictures in a NENO set and their properties related to the non-expandability of the set (see Theorem  \ref{l-non-exp}).
Subsequently, we will present some further conditions on the internal part of the pictures in a NENO set (see Theorem  \ref{l-non-ov}).
The obtained conditions, all together, will be sufficient to construct a NENO set.

Let us introduce the following definition on string sets. 

%
 \begin{definition}\label{d-full}
Let $S_1$, $S_2 \subseteq \Sigma^n$ and $S_1\cap S_2=\emptyset$.
The pair $(S_1,S_2)$ is {\em full} if for any $s\notin S_1\cup S_2$, there exist
$s_1\in S_1$, $s_2\in S_2$ such that $s$ and $s_1$ overlap, and  $s$ and $s_2$ overlap.
 \end{definition}

%

\begin{example}\label{e-righe-fanta}
Let $\Sigma=\{ 0, 1\}$ and let $S_1, S_2\subseteq\Sigma^n$ be  the languages
 $S_1=\{ 1^n\}$ and $S_2=\{ 0w0 \ |\ w\in \{ 0, 1\}^{n-2}\}$.
The pair $(S_1,S_2)$ is  cross-non-overlapping, since the  strings in $S_1$ do not contain occurences of symbol $0$.

Let us show that it is also
full.
Let $s$ be any string $s\notin S_1\cup S_2$.
Since $s\notin S_2$, three cases are possible,
$s=0x1$, $s=1y0$, or $s=1z1$.
If $s=0x1$ or $s=1y0$ then $s$ overlaps the string in $S_1$ and any string in $S_2$ with an overlap of length $1$.
If $s=1z1$ then $z$ contains at least one occurence of $0$, because $s\notin S_1$.
Then $s=1^k0r1$, for some $k\geq 1$ and $r\in\Sigma^*$.
Therefore, $s$ and $1^n$ overlap, and also $s$ and $0r10^k$ overlap, with $0r10^k\in S_2$.
\end{example}

\begin{example}\label{e-colonne-fanta}
Let $\Sigma=\{ 0, 1\}$ and let $S_3, S_4\subseteq\Sigma^m$ be the languages
 $S_3=\{ 110^{m-2}\}$ and $S_4=\{ 1w0 \ |\ |w|=m-2, w\neq 0^{m-2}, 1w0$ with no suffix in $110^+\}$.
The pair $(S_3,S_4)$ is  cross-non-overlapping. Indeed, consider $s'=110^{m-2} \in S_3$ and $s''=1w0 \in S_4$ and suppose that
there exists $u$ such that $s'=xu$ and $s''=uy$ or $s'=ux$ and $s''=yu$. In the first case, it must be $|u|=m-1$ and, hence, $w=0^{m-2}$ against the definition of $S_4$.
In the second case, it must be $|u| \geq 3$ and, hence, $w$ would have a suffix in $110^+$ against the definition of $S_4$.

Let us show that the pair $(S_3,S_4)$ is full.
Let $s$ be any string $s\notin S_3\cup S_4$.
If $s\in 0\Sigma^*$ or $s\in \Sigma^* 1$ then $s$ overlap the string in $S_3$ and any string in $S_4$ with an overlap of length $1$.
Consider the case that $s=1w0$.
Since $s\notin S_4$, $w=0^{m-2}$ or $w$ has a suffix in $110^+$.
In the first case $s=10^{m-1}$. Then $s$ and the string in $S_3$ overlap with an overlap of length $m-1$; and $s$
overlaps any string in $S_4$ that has $010$ as a suffix, with an overlap of length $2$.

In the second case
$s=1x110^k$ with $k\geq 1$, and $x\in\Sigma^*$.
Then $s$ and the string in $S_3$ overlap with an overlap of length $k+2$,
and $s$ and any string $10^k y0\in S_4$, for some $y$, overlap with an overlap of length $k+1$.
\end{example}


Recall that the frame of a language $X$ 
 is the quadruple
$frame(X)$
of the sets of its first and last rows and columns,
$frame(X)=(R_F(X), R_L(X), C_F(X), C_L(X))$.
Note that not any quadruple of string languages  can be the frame of a set of pictures, since their strings need to match    in the lengths and in the corner positions.
Given four string languages $S_1, S_2\subseteq \Sigma^n$ and
 $S_3, S_4\subseteq \Sigma^m$,
we say that the quadruple  $(S_1,S_2,S_3,S_4)$ is {\em frame-compatible} if there exists a picture language $X\subseteq \Sigma^{m,n}$ such that
$frame(X)=(S_1,S_2,S_3,S_4)$.
In the case of non-overlapping sets of pictures,
the constrains are yet stronger.
As previously remarked (see Lemma \ref{l-corner-non-ov-sets}),
if $X$ is a non-overlapping set on $\Sigma=\{ 0, 1\}$, then all the pictures in $X$ have the same set of corners and, for this set, only four cases are possible,
following which symbols appear in the corners of pictures in $X$.
Hence, for example, a quadruple  $(S_1,S_2,S_3,S_4)$ can be the frame of a non-overlapping set $X$ such that, for all $p \in X$, $corners(p)=(0,0,1,1)$,
if $S_1\subseteq 0\Sigma^* 0$,
$S_2\subseteq 1\Sigma^* 1$, and
$S_3, S_4\subseteq 0\Sigma^* 1$.

We give the following definition.


\begin{definition}
Let $X\subseteq \Sigma^{m,n}$.
$X$ is {\em frame-complete} if for any $p, q\in X$, $p$ and $q$ do not frame overlap, and
if for any picture $p \in \Sigma^{m,n}\setminus X$ there exists a picture $q\in X$ such that $p$ and $q$ frame overlap.
\end{definition}

%
Note that being frame-complete is  a \emph{sufficient condition} for a set $X$ to be non-expandable with respect to the overlapping. Any picture not in the language overlaps on the frame of some picture in the language. 
Frame-complete sets of pictures can be obtained thanks to the following theorem.

\begin{theorem}\label{l-non-exp}
Let $S_1, S_2\subseteq \Sigma^{n}$,
$S_3, S_4\subseteq \Sigma^{m}$
and $(S_1, S_2, S_3, S_4)$ be a quadruple of frame-compatible string languages.
\\
If the pairs $(S_1, S_2)$ and $(S_3, S_4)$
are   cross non-overlapping and full
then the set $X$ of all the pictures $p$ with $frame(p)\in S_1 \times S_2\times S_3\times S_4$
is frame-complete.
\end{theorem}

\begin{proof}

The frame-compatibility of $(S_1, S_2, S_3, S_4)$ guarantees that $X$ is not empty.

Moreover, since the pairs $(S_1, S_2)$ and $(S_3, S_4)$
are cross-non-overlapping, then, for any $p,q \in X$, we are sure that $p$ and $q$ do not frame overlap.

Now let $p$ be any picture $p \in \Sigma^{m,n}\setminus X$.
Let us show that there exists a picture $q\in X$ such that $p$ and $q$ frame overlap.
The definition of $X$ implies that $frame(p)\notin  S_1 \times S_2\times S_3\times S_4$.
Hence,
either $r_F(p)\notin S_1$, or
 $r_L(p)\notin S_2$, or
 $c_F(p)\notin S_3$, or
 $c_L(p)\notin S_4$.
Suppose without loss of generality that $r_F(p)\notin S_1$.

If also $r_F(p)\notin S_2$ then there exists $s\in S_2$ such that $r_F(p)$ and $s$ overlap (because $(S_1,S_2)$ is full).
Consider any picture $q$ in $X$ with $r_L(q)=s$. Then, $p$ and $q$ frame overlap.

If $r_F(p)\in S_2$ then
consider any picture $q$ in $X$ with $r_L(q)= r_F(p)$. Then, $p$ and $q$ frame overlap (and also v-slide overlap).
\end{proof}

\begin{example}\label{e-verso-fantastica}
Let $\Sigma=\{ 0, 1\}$. Referring to the sets $S_1$, $S_2$, $S_3$, and $S_4$ in Examples \ref{e-righe-fanta} and
\ref{e-colonne-fanta}, let
 $X(m,n)\subseteq\Sigma^{m,n}$, with $m,n \geq 4$,
be the set of all the pictures with
$R_F(X(m,n))=S_1$, $R_L(X(m,n))=S_2$,
$C_F(X(m,n))=S_3$, and $C_L(X(m,n))=S_4$.
The sets $R_F(X(m,n))$, $R_L(X(m,n))$, $C_F(X(m,n))$, $C_L(X(m,n))$ satisfy the conditions of Theorem \ref{l-non-exp} and therefore
$X(m,n)$ is frame-complete.
\end{example}

%
%
%

The next theorem states some sufficient conditions on a subset of a frame-complete set, so that the subset is NENO.


\begin{theorem}\label{l-non-ov}
Let $X\subseteq \Sigma^{m,n}$ be a frame-complete set.
If a subset $Y$ of $X$ is such that
\begin{enumerate}
\item [a)] $frame(Y)=frame(X)$
\item[b)] $Y$ is non-overlapping
\item[c)] for any $p\in X\setminus Y$  there exists $q\in Y$ such that $p$ and $q$ overlap
\end{enumerate}
then
$Y$ is a NENO set.
\end{theorem}

\begin{proof}
The set $Y$ is non-overlapping by condition b).
Let us show that $Y$ is non-expandable.
Let $p$ be any picture $p\in\Sigma^{m,n}\setminus Y$, and let us show that there exists $q\in Y$ such that $p$ and $q$ overlap.
If $p\in X$ then condition c) implies the goal.
Suppose $p\notin X$.
Then there exists $x\in X$ such that $p$ and $x$ frame overlap, because $X$ is frame-complete.
Condition a) guarantees that there exists  $q\in Y$ such that $frame(q)=frame(x)$ and hence $p$ and $q$ frame overlap, too.
\end{proof}

Previous Theorems \ref{l-non-exp} and \ref{l-non-ov} suggest a procedure to construct a NENO set.

The construction is accomplished in two main steps. 

The first step 
is the construction of
a frame-complete set $X$ of pictures. This can be accomplished by choosing a quadruple
 $(S_1, S_2, S_3, S_4)$  of frame-compatible string languages, such that
 the pairs $(S_1, S_2)$ and $(S_3, S_4)$
are   cross non-overlapping and full,
and then defining
 the set $X$ of all the pictures $p$ with $frame(p)\in S_1 \times S_2\times S_3\times S_4$.
The set $X$
is frame-complete, applying Theorem \ref{l-non-exp}.


The second step consists in selecting  a subset $Y$ of $X$. 
If the choice of $Y$ is done in a way that it satisfies conditions a), b), and c) in Theorem  \ref{l-non-ov},
then Theorem  \ref{l-non-ov}  ensures that $Y$ is a NENO set.

The next section shows a family of NENO sets obtained by following this construction.


\section{A family of NENO sets}\label{s-fanta-neno}
%
%
%
%
%
In this section we exhibit a family of non-expandable non-overlapping sets.
They are obtained following the construction in Section \ref{s-constr}.
 The first step is accomplished in Example \ref{e-verso-fantastica}, where the family of frame-complete sets $X(m,n)$, for $m, n\geq 4$, is presented.
The sets $X(m,n)$ were constructed by taking care of the frames of pictures.
Starting from $X(m,n)$, we will implement the second step of the construction, that is the extraction from $X(m,n)$, of some subset satisfying the 
conditions
in Theorem \ref{l-non-ov}. This time, the focus in on the internal part of the pictures.

Observe that the second step of the construction needs to balance accurately two opposite operations.
On one side, one has to remove
 from $X(m,n)$ those pictures that overlap other pictures in the set.
On the other hand, it is necessary not to remove ``too many" pictures, in order to achieve the non-expandability property.
The following family reaches this goal.

\begin{definition}\label{e-fantastica} 

Let $\Sigma=\{ 0, 1\}$ and  $m,n\geq 4$.

Let $S_1, S_2\subseteq\Sigma^n$, and  $S_3, S_4\subseteq\Sigma^m$, be the languages
 $S_1=\{ 1^n\}$,
 $S_2=\{ 0w0 \ |\ w\in \{ 0, 1\}^{n-2}\}$,
 $S_3=\{ 110^{m-2}\}$,
 and $S_4=\{ 1w0 \ |\ |w|=m-2,  w\neq 0^{m-2}, 1w0$ with no suffix in $110^+\}$.

Let
 $X(m,n)\subseteq\Sigma^{m,n}$
be the set
of all the pictures with
$R_F(X(m,n))=S_1$, $R_L(X(m,n))=S_2$,
$C_F(X(m,n))=S_3$, and $C_L(X(m,n))=S_4$.

The set $Y(m,n)\subseteq X(m,n)$ is the set of all pictures $p\in X(m,n)$ such that
\begin{enumerate}
\item there exists no $2\leq j\leq n$ such that $p(2,j)=p(3,j)=\cdots =p(m-1,j)=0$
\item there exists no $(i,j)$, with $(i,j)\neq (1,1)$ 
such that $p(i,j)=p(i,j+1)=\cdots =p(i,n)=1$ and $p(i,j)p(i+1,j)\cdots p(m,j)\in 110^*$.
\end{enumerate}
\end{definition}

%
%
%
%
%
%

\begin{figure}
\hspace{10pt}
\begin{picture}(170,25)(0,0)
\put(2.5,26.5){$1$} \put(7.5,26.5){$1$} \put(12.5,26.5){$\ldots$} \put(17.5,26.5){$1$} \put(22.5,26.5){$1$}
\put(2.5,21.5){$1$}
\put(2.5,16.5){$0$}
\put(2.5,11.5){$\vdots$}
\put(2.5,6.5){$0$}
\put(22.5,14.5){$x$}
\put(2.5,1.5){$0$} \put(12.5,1.5){$w$} \put(22.5,1.5){$0$}
\put(0,30){\line(1,0){25}}
\put(0,25){\line(1,0){25}}
\put(0,20){\line(1,0){5}}
\put(0,5){\line(1,0){25}}
\put(0,0){\line(1,0){25}}
\put(0,0){\line(0,1){30}}
\put(5,0){\line(0,1){30}}
\put(20,0){\line(0,1){30}}
\put(25,0){\line(0,1){30}}

\put(37.5,26.5){$1$} \put(42.5,26.5){$1$} \put(47.5,26.5){$1$} \put(52.5,26.5){$1$} \put(57.5,26.5){$1$}
\put(37.5,21.5){$1$} \put(42.5,21.5){$0$} \put(47.5,21.5){$1$} \put(52.5,21.5){$0$} \put(57.5,21.5){$1$}
\put(37.5,16.5){$0$} \put(42.5,16.5){$0$} \put(47.5,16.5){$1$} \put(52.5,16.5){$0$} \put(57.5,16.5){$1$}
\put(37.5,11.5){$0$} \put(42.5,11.5){$1$} \put(47.5,11.5){$1$} \put(52.5,11.5){$1$} \put(57.5,11.5){$0$}
\put(37.5,6.5){$0$}  \put(42.5,6.5){$1$} \put(47.5,6.5){$1$} \put(52.5,6.5){$0$} \put(57.5,6.5){$1$}
\put(37.5,1.5){$0$}  \put(42.5,1.5){$1$} \put(47.5,1.5){$0$} \put(52.5,1.5){$0$} \put(57.5,1.5){$0$}
\put(35,30){\line(1,0){25}}
\put(35,25){\line(1,0){25}}
\put(35,20){\line(1,0){25}}
\put(35,15){\line(1,0){25}}
\put(35,10){\line(1,0){25}}
\put(35,5){\line(1,0){25}}
\put(35,0){\line(1,0){25}}
\put(35,0){\line(0,1){30}}
\put(40,0){\line(0,1){30}}
\put(45,0){\line(0,1){30}}
\put(50,0){\line(0,1){30}}
\put(55,0){\line(0,1){30}}
\put(60,0){\line(0,1){30}}

\put(72.5,26.5){$1$} \put(77.5,26.5){$1$} \put(82.5,26.5){$1$} \put(87.5,26.5){$1$} \put(92.5,26.5){$1$}
\put(72.5,21.5){$1$}                      \put(82.5,21.5){$0$}                      \put(92.5,21.5){$1$}
\put(72.5,16.5){$0$}                      \put(82.5,16.5){$0$}                      \put(92.5,16.5){$1$}
\put(72.5,11.5){$0$}                      \put(82.5,11.5){$0$}                      \put(92.5,11.5){$0$}
\put(72.5,6.5){$0$}                       \put(82.5,6.5){$0$}                      \put(92.5,6.5){$1$}
\put(72.5,1.5){$0$}  \put(77.5,1.5){$1$}  \put(82.5,1.5){$0$}  \put(87.5,1.5){$0$}  \put(92.5,1.5){$0$}
\put(70,30){\line(1,0){25}}
\put(70,25){\line(1,0){25}}
\put(70,20){\line(1,0){25}}
\put(70,15){\line(1,0){25}}
\put(70,10){\line(1,0){25}}
\put(70,5){\line(1,0){25}}
\put(70,0){\line(1,0){25}}
\put(70,0){\line(0,1){30}}
\put(75,0){\line(0,1){30}}
\put(80,0){\line(0,1){30}}
\put(85,0){\line(0,1){30}}
\put(90,0){\line(0,1){30}}
\put(95,0){\line(0,1){30}}

\put(107.5,26.5){$1$} \put(112.5,26.5){$1$} \put(117.5,26.5){$1$} \put(122.5,26.5){$1$} \put(127.5,26.5){$1$}
\put(107.5,21.5){$1$}                                                                   \put(127.5,21.5){$1$}
\put(107.5,16.5){$0$}                       \put(117.5,16.5){$1$} \put(122.5,16.5){$1$} \put(127.5,16.5){$1$}
\put(107.5,11.5){$0$}                       \put(117.5,11.5){$1$}                       \put(127.5,11.5){$0$}
\put(107.5,6.5){$0$}                        \put(117.5,6.5){$0$}                        \put(127.5,6.5){$1$}
\put(107.5,1.5){$0$}  \put(112.5,1.5){$1$}  \put(117.5,1.5){$0$}  \put(122.5,1.5){$0$}  \put(127.5,1.5){$0$}
\put(105,30){\line(1,0){25}}
\put(105,25){\line(1,0){25}}
\put(105,20){\line(1,0){25}}
\put(105,15){\line(1,0){25}}
\put(105,10){\line(1,0){25}}
\put(105,5){\line(1,0){25}}
\put(105,0){\line(1,0){25}}
\put(105,0){\line(0,1){30}}
\put(110,0){\line(0,1){30}}
\put(115,0){\line(0,1){30}}
\put(120,0){\line(0,1){30}}
\put(125,0){\line(0,1){30}}
\put(130,0){\line(0,1){30}}

\end{picture}
\caption{}
\label{fig-fantastica}
\end{figure}

Figure \ref{fig-fantastica} shows 
 a generic picture in $X(m,n)$, a picture in $Y(m,n)$, and two pictures which are in $X(m,n)$, but are not in $Y(m,n)$ because they violate condition 1 and 2, respectively.

Note that  $X(m,n)\subseteq\Sigma^{m,n}$ in the definition
is the frame-complete set introduced  in Example \ref{e-verso-fantastica}.
The set $Y(m,n)$ has been extracted from $X(m,n)$ in such a way that it  satisfies the sufficient conditions in Theorem \ref{l-non-ov}.
More precisely, the conditions 1 and 2 in the definition of $Y(m,n)$ are designed in order to avoid the overlaps between two pictures in $Y(m,n)$.
In particular, condition 1 avoids bl-overlaps.
The bl-corner of a picture $p$ in $X(m,n)$ cannot be placed inside another picture $p'$ of $X(m,n)$, unless when it is placed on the $(m-1)$-th row (in all the other cases the $0$'s in the first column of $p$ would meet the $1$'s in the first row of $p'$).
Condition 1 avoids  this possibility in the pictures in $Y(m,n)$.
In a similar way, condition 2 forbids tl-overlaps.
On the contrary, the pictures in $X(m,n)$ which do not satisfy conditions 1 and 2, i. e. the pictures in $X(m,n)\setminus Y(m,n)$, will necessarily overlap some pictures in $Y(m,n)$.

The following theorem shows, in a more accurate way, that $Y(m,n)$ satisfies the conditions a), b), and c) of Theorem \ref{l-non-ov}, and then $Y(m,n)$ is a NENO set. 

\begin{theorem}\label{p-fanta-e-neno}
The language $Y(m,n)$ in Definition \ref{e-fantastica} is a NENO set, for any $m,n\geq 4$.
\end{theorem}

\begin{proof}
Let $X(m,n)$ and $Y(m,n)$ be the languages defined in Definition \ref{e-fantastica}.
The set $X(m,n)$ is frame-complete, as shown in Example \ref{e-verso-fantastica}.
Let us show that $Y(m,n)$ satisfies the conditions a), b), and c), of Theorem \ref{l-non-ov}.

\begin{enumerate}
\item [a)]  {\em $frame(Y(m,n))=frame(X(m,n))$ }

Since $Y(m,n)\subseteq X(m,n)$ then $frame(Y(m,n))\subseteq frame(X(m,n))$.
Let us show the inverse inclusion, namely, that for any $p\in X(m,n)$ there exists $q\in Y(m,n)$ and $frame(q)=frame(p)$.
It is trivial if $p\in Y(m,n)$. Suppose $p\in X(m,n)\setminus Y(m,n)$.
Then, $p$ does not satisfy condition 1 or condition 2 (or both) in Definition \ref{e-fantastica}.
It is possible to obtain a picture $q\in Y(m,n)$ with $frame(q)=frame(p)$, by exchanging some symbols in $p$.
More precisely,
for any $1\leq j\leq n$ such that $p(2,j)=p(3,j)=\cdots =p(m-1,j)=0$,
one can  replace with $1$ one occurrence  of $0$ among $p(3,j),\cdots , p(m-1,j)$.
This replacement reduces the violations of condition 1 and does not introduce any violation to condition 2.

Suppose now that there exists $(i,j)$, with $(i,j)\neq (1,1)$ 
such that $p(i,j)=p(i,j+1)=\cdots =p(i,n)=1$ and $p(i,j)p(i+1,j)\cdots p(m,j)\in 110^*$.
Recalling that $m,n\geq 4$, there always exists one position among
$(i,j)$, $(i,j+1), \cdots, (i,n)$ and $(i+1,j), \cdots , (m,j)$
 that is not on the frame of $p$.
Exchanging the symbol in such position removes the violation
and does not introduce any violation to condition 1.

\item[b)] {\em $Y(m,n)$ is non-overlapping.}

Consider $p, q \in Y(m,n)$.
The pictures $p$ and $q$ cannot frame overlap, because the pairs
$(R_F(X(m,n))$, $R_L(X(m,n)))$, $(C_F(X(m,n))$, $C_L(X(m,n)))$ are cross-non-overlapping (as shown in Examples \ref{e-righe-fanta} and
\ref{e-colonne-fanta}).

The pictures $p$ and $q$ cannot h-slide overlap, because otherwise it would exist a position $(1,j)$ with $(i,j)\neq (1,1)$ violating condition 2.

The pictures $p$ and $q$ cannot v-slide overlap because the first column of any picture in $Y(m,n)$ is unbordered.

Moreover, $p$ and $q$ cannot tl-overlap with an overlap of size  $(r,c)$ with $1 < r\leq m-1$
and $1 < c\leq  n-1$.
Suppose by the contrary that there exists $x$ of size $(r,c)$ with $1 < r\leq m-1$
and $1 < c\leq  n-1$, such that $x$ is a tl-prefix of $p$ and a br-prefix of $q$.


Finally, $p$ and $q$ cannot bl-overlap with an overlap of size  $(r,c)$ with $1 < r\leq m-1$
and $1 < c\leq  n-1$.
Suppose by the contrary that there exists $x$ of size $(r,c)$ with $1 < r\leq m-1$
and $1 < c\leq  n-1$, such that $x$ is a bl-prefix of $p$ and a tr-prefix of $q$.

Since $r_F(q)=1^n$, $r$ cannot be strictly less than $m-1$.
On the other hand, if $r=m-1$ then the picture $q$ violates condition 1 in the column $j=n-c+1$.

\item[c)] {\em For any $p\in X(m,n)\setminus Y(m,n)$  there exists $q\in Y(m,n)$ such that $p$ and $q$ overlap.}

 If $p\in X(m,n)\setminus Y(m,n)$ then  $frame(p)\in S_1 \times S_2\times S_3\times S_4$, but
$p$ does not satisfy condition 1 or condition 2 (or both) in Definition \ref{e-fantastica}.

If $p$ does not satisfy condition 1 then let $j$ be the greatest index $2\leq j\leq n$ such that $p(2,j)=p(3,j)=\cdots = p(m-1,j)=0$.
Then there exists $q\in Y(m,n)$ with the bl-prefix of $q$ of size $(m-1, n-j+1)$ equal to the tr-prefix of $p$ of size $(m-1, n-j+1)$.
Note that since $j\neq 1$  the last column of $q$ can be constructed so that no violation of condition 2 appears in $q$, by inserting some occurrences of $0$'s where necessary. Furthermore, such $0$'s are not necessary in each position $q(2,n)$, $q(3,n)$, $\cdots, q(m-1,n)$ since the last column of $p$ cannot be in $1^+0^+$.

If $p$ satisfies condition 1, but not condition 2, 
then
let $(i,j)$ the lowest among the rightmost positions
such that $p(i,j)=p(i,j+1)=\cdots =p(i,n)=1$ and $p(i,j)p(i+1,j)\cdots p(m,j)\in 110^*$.
Then there exists $q\in Y(m,n)$ with the tl-prefix of $q$ of size $(m-i+1, n-j+1)$ equal to the br-prefix of $p$ of size $(m-i+1, n-j+1)$.
\end{enumerate}
\end{proof}

We conclude the section with some remarks. Following the construction presented in Section \ref{s-constr}, for any $m,n\geq 4$,
we have extracted from the frame-complete set $X(m,n)$ a NENO set $Y(m,n)$, by imposing two conditions on the internal part of its pictures.
Such conditions are designed so that the pictures in $X(m,n)$ that satisfy them cannot overlap each other, whereas as soon as a picture in $X(m,n)$ does not satisfy a condition, it necessarily overlaps a picture in $Y(m,n)$.
Let us emphasize that not any conceivable 
 condition which avoids any overlap when satisfied, will necessarily imply a desired overlap when it is not satisfied.

As an example, consider the language $M(m,n)\subseteq X(m,n)$, defined by replacing in the definition of $Y(m,n)$, condition 1 with the following condition 1 bis.
\begin{description}
\item[1bis.] if  there exists  $(i,j)$, with $(i,j)\neq (1,n)$
such that $p(i,1)=p(i,2)=\cdots =p(i,j)=1$ then  $p(i,j)p(i+1,j)\cdots p(m,j)$    has a suffix in $110^*$.
\end{description}
Condition 1bis seems to avoid bl-overlaps in a similar way as condition 1 does. Nevertheless,
one can show that $M(m,n)$ is non-overlapping,  but it is not non-expandable.
A counter-example, for $m=n=5$, is the  picture 
$p=
\begin{array}{|ccccc|} \hline
1& 1& 1 & 1 & 1\\
1& 0& 1 & 1 & 1\\
0&  1& 0 & 1 & 0\\
0&  1& 0 & 0 & 1\\
0&  1& 0 & 1 & 0\\
\hline
\end{array}\, .$
The picture $p$ belongs to $X(m,n)\setminus M(m,n)$, but there is no picture $q\in X$ such that $p$ and $q$ overlap.

%
%


\section{Cardinality of non-overlapping sets of pictures}\label{s-counting}

%

In the previous sections we have discussed the construction of NENO sets and presented a family of NENO sets.
We consider now the question of how big is this family and how big a NENO set can be.
The question has been extensively investigated in the string case.
Let $C(n,q)$ denote the maximum size of a cross-bifix-free code of strings of length $n$ over an alphabet of cardinality $q$.
In \cite{CKPW13},
it is shown that
$C(n,q)\leq \frac{q^n}{2n-1}.$
The non-expandable cross-bifix-free sets of strings  introduced in \cite{Bajic14,BPP12,CKPW13} have been compared with this bound.

Let us apply the upper bound on strings, in order to obtain a simple upper bound on the size of a non-overlapping set of pictures.
Let $C(m,n,q)$ denote the maximum size of a non-overlapping set of pictures of size $(m,n)$ over an alphabet of cardinality $q$.

\begin{theorem}
Let $m,n, q$ be integers and $N=\max\{m,n\}$. Then
\\
$$C(m,n,q)\leq \frac{q^{mn}}{2N-1}.$$
\end{theorem}
\begin{proof}
Let $X\subseteq \Sigma^{m,n}$ and $|\Sigma|=q$.
Note that
$X$ can be viewed as a set $col(X)$ of strings of length $n$ on the alphabet of the columns $\Sigma^m$, and  $|\Sigma^m| =q^m$. 
The key observation is that if $X$ is non-overlapping then $col(X)$ is cross-bifix-free.
Actually, an overlapping of two strings in $col(X)$ would give a h-slide-overlap of two pictures in $X$.
Therefore, applying the upper bound given in \cite{CKPW13} on the cardinality of the cross-bifix-free sets of strings,
$|X|\leq \frac{ (q^{m})^n} {2n-1}.$
In an analogous manner,
$X$ can be viewed as a set $row(X)$ of strings of length $m$ on the alphabet of the rows $\Sigma^n$,  and  $|\Sigma^n| =q^n$.
Hence,
$|X|\leq \frac{ (q^{n})^m} {2m-1}.$
Finally,
$|X|\leq \min \{ \frac{ q^{mn}} {2n-1} , \frac{ q^{mn}} {2m-1} \} = \frac{ q^{mn}} {2N-1}.$
\end{proof}


Let us now evaluate the cardinality of the NENO sets $Y(m,n)$ introduced in 
Definition \ref{e-fantastica}.
We are going to present a lower bound on the cardinality of $Y(m,n)$.

Let $\Sigma=\{ 0, 1\}$.
Define the set $Z(m,n)$ of all the pictures $z\in \Sigma^{m,n}$ with
 $frame(z)\in frame(Y(m,n))$
and such that
\begin{description}
\item[1.] there exists no $2\leq j\leq n$ such that $p(2,j)=p(3,j)=\cdots =p(m-1,j)=0$
\item[2a.] there exists no $(i,j)$, with $(i,j)\neq (1,1)$ 
such that
$p(i,j)p(i+1,j)\cdots p(m,j)\in 110^*$.
\end{description}

Observe that
$Z(m,n)\subseteq Y(m,n)$, since the condition 2a implies condition 2 in the definition of $Y(m,n)$.
Therefore, $|Y(m,n)| \geq |Z(m,n)|$.
In the next theorem we will evaluate $|Z(m,n)|$, obtaining a lower bound on $|Y(m,n)|$.

\begin{theorem}
Let $Y(m,n)$ the NENO set in Definition \ref{e-fantastica}.
Then
 $$|Y(m,n)| \geq \bigl( 2^{m-2} -1 \bigr ) ^{n-2} \cdot \  2^{m-3}$$
\end{theorem}
\begin{proof}
Let $\Sigma=\{ 0,1\}$ and $Z(m,n)\subseteq \Sigma^{m,n}$ be the set defined above.
We prove that $|Z(m,n)|= \bigl( 2^{m-2} -1 \bigr ) ^{n-2} \cdot \  2^{m-3}$. This shows the statement since $Z(m,n)\subseteq Y(m,n)$ .

Define the sets
$I(m), L(m)\subseteq \Sigma^m$ as follows.

$I(m)=\{  1y a \ | \ a\in\Sigma, y\neq 0^{m-2}$ and $1ya$ with no suffix in $110^* \}$

$L(m)=\{  1x0 \ | \  x\neq 0^{m-2}$ and $1x0$ with no suffix in $110^+ \}$.

Consider any picture $z\in Z(m,n)$.
The first column of $z$ is $110^{m-2}$;
the internal columns of $z$, that is the columns with index $2\leq j\leq n-1$, are all strings in $I(m)$, while
the last column of $z$ is a string in $L(m)$.

Let us show that $|I(m)| = 2^{m-2}-1$.
The strings $1ya$ of length $m$ with a suffix in $110^+$ are $\sum_{i=0}^{m-3} 2^i = 2^{m-2}-1$, while
the  strings $1ya$ of length $m$ with $y=0^{m-2}$ are $2$.
Therefore,
$|I(m)|= 2^{m-1} - \bigl( 2^{m-2}-1 \bigr) - 2 = 2^{m-2}-1$.

Let us show that $|L(m)| = 2^{m-3}$.
The strings $1x0$ of length $m$ with a suffix in $110^*$ are $\sum_{i=0}^{m-4} 2^i = 2^{m-3}-1$, while
the  strings $1x0$ of length $m$ with $x=0^{m-2}$ is just $1$.
Therefore,
$|L(m)|= 2^{m-2} - \bigl( 2^{m-3}-1 \bigr) - 1 = 2^{m-3}$.

Any picture $z\in Z(m,n)$ is just the catenation of the column $110^{m-2}$ with $n-2$ columns in $I(m)$ and then a column in $L(m)$, with no other constrain.
Hence
 $|Z(m,n)|= |I(m)|^{n-2} \cdot |L(m)|=
\bigl( 2^{m-2} -1 \bigr ) ^{n-2} \cdot \  2^{m-3}$.
\end{proof}

\bibliographystyle{splncs03}
\bibliography{bibliography2}

\end{document}